\documentclass[review,a4paper,11pt]{article}
\usepackage{amsmath, amsthm, amssymb,enumerate}

\usepackage{natbib}
\usepackage{tikz}
\usetikzlibrary{shapes,arrows,fit,calc,positioning,patterns}

\newtheorem{theorem}{Theorem}
\newtheorem{lemma}{Lemma}
\newtheorem{claim}{Claim}
\newtheorem{prop}{Proposition}

\newtheorem{corollary}{Corollary}

\newcommand{\J}{\mathcal{J}}
\newcommand{\B}{\mathcal{B}}

\newcommand{\lmp}{late work minimization problem}
\newcommand{\rlp}{resource leveling problem}
\newcommand{\emp}{early work maximization problem}

\title{A common approximation framework for the early work, the late work, and  resource leveling problems with unit time jobs\footnote{This work has been supported by the National Research, Development and Innovation Office -- NKFIH, grant no.~SNN~129178, and ED\_18-2-2018-0006.}}

\author{P\'eter Gy\"orgyi and Tam\'as Kis\footnote{Corresponding author.}\\
Institute for Computer Science and Control, Kende str 13-17,\\
 1111 Budapest, Hungary\\
peter.gyorgyi@sztaki.hu and  tamas.kis@sztaki.hu}

\begin{document}
\maketitle

\begin{abstract}
We study the approximability of two related machine scheduling problems.
In the {\em late work minimization problem\/}, there are identical  parallel machines and the jobs have a common due date. The objective is to minimize the {\em late work\/}, defined as the sum of the portion of the jobs done after the due date.
A related problem is the  maximization of the {\em early work\/}, defined as the sum of the portion of the jobs done before the due date.
We describe a polynomial time approximation scheme for the early work maximization problem, and we extended it to the late work minimization problem after shifting the objective function by a positive value that depends on the problem data.
We also prove an  inapproximability result for the latter  problem if the objective function is shifted by a constant which does not depend on the input.
These results remain valid even if the number of the jobs assigned to the same machine is bounded. This  leads to an extension of our approximation scheme to some variants of the resource leveling problem, for which no approximation algorithms were known.
\end{abstract}

\section{Introduction}\label{sec:intro}
Late work minimization, introduced by the pioneering paper of \citet{blazewicz1984scheduling}, is an important area of machine scheduling, for an overview see \citet{sterna2011survey}.
The variant we are going to study in this paper can be briefly stated as follows. We have identical parallel machines and a set of jobs with a common due date. We seek a schedule which minimizes the sum of the portion of the jobs done after the due date.
A strongly related problem is the maximization of the early work, where we have the same data and the objective is to maximize the sum of the portion of the jobs done before the common due date. However, the list of the results for maximizing the early work is much shorter than that for the late work minimization problem, see e.g., \citet{sterna2017polynomial}. 

The applications of the late work optimization criterion range from modeling the loss of information in computational tasks to the measurement of dissatisfaction of the customers  of a manufacturing company.
In particular,  \citet{blazewicz1984scheduling} studies a parallel processor scheduling problem with preemptive jobs where each job processes some samples of data (or measurement points), and if the processing completes after the job's due date, then it causes a loss of information. A natural objective is to minimize the information loss, which is equivalent to the minimization of the  total late work. 
A small flexible manufacturing system is described in \citet{sterna2007late}, where the application of the late work criterion is motivated by the interests of the customers as well as by that of the owner of the system.
The common interest of the customers is to  have the portions of their orders finished after the due date minimized.  In turn, for the owner of the system, the amount of late work is a measure of dissatisfaction of the customers.
As for the early work, we can adapt the same examples considering gain and satisfaction instead of loss and dissatisfaction, respectively.

We have three major sources of motivation for studying the approximability of the early work maximization, and the late work minimization problems:
\begin{enumerate}[i)]
\item \citet{chen2016scheduling} establish the complexity of late work minimization in a parallel machine environment, and then the authors describe  an online algorithm for the early work maximization problem of competitive ratio $\frac{\sqrt{2m^2-2m+1}-1}{m-1}$. However, since the late work can be 0, no approximation or online algorithm is proposed for the late work objective.
\item \citet{sterna2017polynomial} propose a polynomial time approximation scheme for the early work maximization problem with 2 machines, and it is not obvious how to get rid of  some constant bound.

\item We have observed that some variants of the resource leveling problem are equivalent to the early work maximization and the late work minimization problems. Briefly, the resource leveling problems we are referring to consist of a parallel machine environment and one more renewable resource required by a set of  unit time jobs having a common deadline, and one aims at to minimize (maximize) the total resource usage above (below) a threshold. We are not aware of any published approximation algorithms for resource leveling problems in a parallel machine environment, but the results for the early- and late work problems can be transferred to this important subclass.
\end{enumerate}

In this paper we propose a common approximation framework for the early work maximization, the late work minimization and the resource leveling problem in a parallel machine environment with unit time jobs. We emphasize that the number of identical parallel machines is part of the input for all problems studied, and the processing times of the jobs are arbitrary positive integer numbers in the early- and late work problems.

\subsection{Problem formulations and terminology}
In the \lmp~in a parallel machine environment, there is a set $\J$ of $n$ jobs that have to be scheduled on $m$ identical parallel machines. 
If it is not noted otherwise,  the number of the machines is part of the input.
Each job $j \in \J$ has a processing time $p_j$ and there is a common due date $d$.
The  late work objective $Y$ is to minimize the total amount of work scheduled after $d$, see  \citet{chen2016scheduling}.
That is, a schedule $S$ specifies a machine $\mu_j(S) \in \{1,\dots,m\}$ and a starting time $t_j(S)\geq 0$ for each job.
$S$ is {\em feasible\/} if  for each pair of distinct jobs $j$ and $k$ such that $\mu_j(S) = \mu_k(S)$, either $t_j(S)+p_j\leq t_k(S)$ or $t_k(S) + p_k \leq t_j(S)$. Throughout the paper we assume that there are no idle times between the jobs on any machine.
The {\em late work\/} of a schedule $S$ is $Y = \sum_{i=1}^m \max\{0, \sum_{j \in J_i(S)} p_j - d\}$, where $J_i(S) = \{ j\in \J \ |\ \mu_j(S) = i\}$. 
Later we will frequently refer to the sum of the job processing times $p_{sum}:= \sum_{j\in\J} p_j$.
  
We add a further constraint to this problem. We introduce a bound $N$ on the number of the jobs that can be scheduled on any of the machines.  This is called {\em machine capacity\/}, see e.g.~\citet{woeginger2005comment}.
Throughout the paper we assume that $m\cdot N\geq n$, otherwise there is no feasible solution for the problem.
Note that machine capacity  is not a common constraint for the \lmp, but it will be useful later. However, by setting $N = n$, the capacity constraints become void, and we get back the familiar \lmp.

Since the late work objective can be $0$, and deciding whether a feasible schedule of late work 0 exists or not is a strongly NP-hard decision problem (\citet{chen2016scheduling}), no approximation algorithm exists for this objective. However, by applying a standard trick, we can ensure that  the objective function value is always positive, and approximating it becomes possible.
We introduce a problem instance-dependent positive number $T$, and when approximating the optimum late work, we will consider the objective function $T+Y$.

There is another way to modify the objective function so that it allows us to achieve approximation results.
The early work objective $X$, introduced by \cite{blazewicz2005two}, which measures the total amount of work scheduled on the machines before $d$, is closely related to $Y$ by the equation
\begin{equation}
X = p_{sum} - Y. \label{eq:lateXY}
\end{equation} 

In the {\em \rlp\/}, we have $n$ jobs with unit processing times to be scheduled on $m$ identical parallel machines in the interval $[0,C]$, where $C$ is a common deadline of all the jobs. Additionally, there is a renewable resource from which $L$ units is available at any time.
Each job $j$ has a resource requirement $a_j\geq 0$ from the resource. All problem data is integral.
 A {\em schedule\/} $S$ specifies a machine $\mu_j(S) \in \{1,\dots,m\}$ and  {\em starting time\/} $t_j(S) \in\{0,\dots,C-1\}$ for each job $j$.
 Without loss of generality, $m \cdot C \geq n$, otherwise no feasible schedule exists.
 Throughout the paper we assume that in any schedule, if $k < m$ jobs start at some time point $t$, then they occupy the first $k$ machines.
The goal is to find a feasible schedule $S$, where each job starts in $[0,C-1]$ and the total resource requirement above $L$ is minimized, i.e., we have to minimize $\tilde{Y}(S) := \sum_{t=0}^{C-1} \max\{0,\sum_{j\in J_t(S)} a_j - L\}$, where $J_t(S) = \{ j \in \J\ |\ t_j(S) = t\}$.
A closely related problem is the  maximization of the total resource usage below $L$ over the scheduling horizon $[0,C]$, i.e., maximize $\tilde{X}(S) := \sum_{t=0}^{C-1} \min\{L, \sum_{j\in J_t(S)} a_j\}$. Let $a_{sum} := \sum_{j \in \J} a_j$.
The two objective functions are related by the equation
\begin{equation}
\tilde{X} = a_{sum}- \tilde{Y}. \label{eq:levelXY}
\end{equation}
Notice the similarity of (\ref{eq:lateXY}) and (\ref{eq:levelXY}). As we will see, this is not a coincidence.
Furthermore, since checking whether a feasible schedule with $\tilde{Y} = 0$ exists is a strongly NP-hard decision problem (\citet{neumann2000resourceleveling}), for approximating the optimal solution we will use the objective function $\tilde{T} + \tilde{Y}$, where 
$\tilde{T}$ is an instance-dependent positive number. 
If $m\geq n$, then we get  the project scheduling version of the resource leveling problem, i.e., there are no machines and arbitrary number of jobs can be started at the same time.

This paper uses the $\alpha|\beta|\gamma$ notation of \citet{graham1979optimization}, where  $\alpha$ denotes the machine environment, $\beta$ the additional constraints, and $\gamma$ the objective function. 
In the $\alpha$ field we use $P$ for arbitrary number of parallel machines and $P2$ in case of two machines.
In the $\beta$ field, $d_j=d$ indicates that the jobs have a common due date, while $n_i \leq N$ indicates the capacity constraints of the machines.
The symbols $X$ and $Y$ in the $\gamma$ field refer to the early work, and to the late work criterion, respectively, and we use the  symbols $\tilde{X}$ and $\tilde{Y}$ to denote the total resource usage below and above the limit $L$, respectively, in case of the resource leveling problem. 

In this paper we describe approximation algorithms for the above mentioned, and some other combinatorial optimization problems.
Our terminology closely follows that of  \citet{garey1979computers}.
 A {\em minimization\/} (resp.~{\em maximization\/}) problem $\Pi$ is given by a set of instances $\mathcal{I}$, and each instance $I \in \mathcal{I}$ has a set of solutions $\mathcal{S}^I$, and an objective function $c^I : \mathcal{S}^I \rightarrow \mathbb{Q}$. Given any instance $I$, the goal is to find a feasible solution $s^* \in \mathcal{S}^I$ such that $c^I(s^*) = \min\{c^I(s)\ |\ s \in \mathcal{S}^I\}$ ($c^I(s^*) = \max\{c^I(s)\ |\ s \in \mathcal{S}^I\}$). Let $OPT(I)$ denote the optimum objective function value of problem instance $I$. A {\em factor $\rho$ approximation algorithm\/} for a minimization (maximization) problem $\Pi$ is a polynomial time algorithm $A$ such that the objective function value, denoted by $A(I)$, of the solution found by the algorithm $A$ on any problem instance $I \in \mathcal{I}$ satisfies $A(I) \leq \rho \cdot OPT(I)$ ($A(I) \geq \rho \cdot OPT(I)$). Naturally, $\rho \geq 1$ for minimization problems, and $0<\rho \leq 1$ for maximization problems.
Furthermore, a {\em polynomial time approximation scheme (PTAS)\/} for $\Pi$ is a family of algorithms $\{A_\varepsilon\}_{\varepsilon>0}$ such that $A_\varepsilon$ is a factor $1+\varepsilon$ approximation algorithm for $\Pi$ if it is a minimization problem, or a factor $1-\varepsilon$ approximation algorithm for $\Pi$ if it is a maximization problem.
In addition, a {\em fully polynomial time approximation scheme (FPTAS)\/} is like a PTAS, but the time complexity of each  $A_\varepsilon$ must be polynomial in $1/\varepsilon$ as well. 
\subsection{Previous work}
In this section first we overview existing complexity and approximability results for scheduling problems with the total late work minimization-, and the total early work maximization objective functions, but we abandon exact and heuristic methods as they are not directly related to our work. Then we briefly overview what is known about resource leveling in a parallel machine environment.

The total late work objective function ({\em late work\/} for short) is proposed by \citet{blazewicz1984scheduling}, where the complexity  of  minimizing the total late work in a parallel machine environment is investigated. For non-preemptive jobs it is mentioned that minimizing the late work is NP-hard, while for preemptive jobs, a polynomial-time algorithm, based on network flows, is described. This approach is extended to uniform machines as well.
Subsequently, several papers have appeared discussing the late work minimization problem in various processing environments. For the single machine environment, \citet{potts1992single} describe an $O(n \log n)$ time algorithm for the problem with preemptive jobs, where each job has its own due date. Furthermore, the non-preemptive variant  is shown to be NP-hard, and among other results, a pseudo-polynomial time algorithm is proposed for finding optimal solutions.
\citet{potts1992approximation} devise a fully polynomial time approximation scheme  for the single machine non-preemptive \lmp, which is extended to the total weighted late work problem by \citet{kovalyov1994fully}, where the late work of each job is weighted by a job-specific positive number.
For a two-machine flow shop, \citet{blazewicz2005two} prove that the late work minimization problem is NP-hard even if all the jobs have a common due date, and they also describe a dynamic programming based exact algorithm.
A more complicated dynamic program is proposed for the two-machine job shop problem with the late work criterion by \citet{blazewicz2007note}.
Late work minimization in an open shop environment, with preemptive or with non-preemptive jobs, is studied in \citet{blazewicz2004open}, where a number of complexity results are proved.
For the parallel machine environment, \citet{chen2016scheduling} prove that deciding whether a schedule with 0 late work exists is a strongly NP-hard decision problem, while if the number of machines is only 2, then it is binary NP-hard.
Furthermore, they describe an online algorithm for maximizing the early work of jobs that have to be scheduled in a given order. For several other complexity results not mentioned here, we refer to \cite{sterna2000problems, sterna2006late, sterna2011survey}.

As for the early work, besides the paper of \citet{chen2016scheduling}, we mention \citet{sterna2017polynomial}, where a PTAS is proposed for maximizing the early work in a parallel machine environment with 2 machines, where  all the jobs have a common due date.

Resource leveling is a well studied area of project scheduling, where a number of exact and heuristic methods are proposed for solving it for various objective functions and under various assumptions, see e.g., \citep{kis2005branch, neumann2000resourceleveling, verbeeck2017metaheuristic}. \citet{drotos2011resource}  consider a dedicated parallel machine environment, and propose and exact method for solving resource leveling problems optimally with hundreds of jobs.
In the same paper, some new complexity results are obtained.

\subsection{Results of the paper}
Before stating our first result, we formally define what we mean by the equivalence of two optimization problems in this paper.
Let $\Pi_1$ and $\Pi_2$ be two optimization problems, and we say that they are {\em equivalent\/} if there exist bijective functions $f$ and $g$, where $f$ establishes a one-to-one correspondence between the instances of $\Pi_1$ and that of $\Pi_2$, whereas $g$ establishes a one-to-one correspondence between the set of solutions of each instance $I$ of $\Pi_1$ and that of $f(I)$ of $\Pi_2$ such that for each $S \in {\cal S}^I$, $c^I(S) = c^{f(I)}(g(S))$\footnote{This is a rather strong concept of equivalence.}. After these preliminaries, we can state our first result.

\begin{theorem}\label{thm:reduction}
The \lmp~$P|d_j=d,n_i\leq N|Y$, and the \rlp~$P|p_j=1|\tilde{Y}$ are equivalent. 
\end{theorem}

By (\ref{eq:lateXY}) and (\ref{eq:levelXY}), we have the following:
\begin{corollary}
The \emp~$P|d_j=d,n_i\leq N|X$, and the  \rlp~$P|p_j=1|\tilde{X}$ are equivalent. \label{cor:earlyequiv}
\end{corollary}

Now we turn to approximation algorithms. 
In Section \ref{sec:inapprox} we show that if we simply add a value $c'$ to $Y$ in the objective function, where $c'$ is an arbitrary positive  number, then it is impossible to get an approximation algorithm of factor smaller  than $\frac{c'+1}{c'}$ unless $P=NP$.

\begin{theorem}\label{thm:inapprox}
For any $\varepsilon>0$, there is no $\left(\frac{c'+1}{c'}-\varepsilon\right)$-approximation algorithm for $P2|d_j=d|c' +Y$ unless $P=NP$.
\end{theorem}

In Section~\ref{sec:earlyPTAS} we describe a PTAS for the  early work maximization problem extended with machine capacity constraints.
\begin{theorem}\label{thm:ptas_ew}
There is a PTAS for $P|d_j=d,n_i \leq N|X$.
\end{theorem}
Since our result is valid even if $N \geq n$, we have the following corollary:
\begin{corollary}
There is a PTAS for $P|d_j=d|X$.
\end{corollary}

By Corollary \ref{cor:earlyequiv}, we immediately get an analogous result for the maximization variant of \rlp:
\begin{corollary}
There is a PTAS for the \rlp~$P|p_j=1|\tilde{X}$.
\end{corollary}

Let  $c > 0$ be any constant, independent of the problem instances of $P|d_j=d,n_i \leq N|Y$.
In Section~\ref{sec:ptas} we adapt the results of Section~\ref{sec:earlyPTAS} to the problem $P|d_j=d,n_i \leq N|c\cdot p_{sum} + Y$. The choice of shifting $Y$ by $c\cdot p_{sum}$ is justified to some extent by Theorem~\ref{thm:inapprox}.

\begin{theorem}\label{thm:ptas}
There is a PTAS for $P|d_j=d,n_i \leq N|c \cdot p_{sum}+Y$.
\end{theorem}
Since our result is valid even if $N \geq n$, we have the following corollary:
\begin{corollary}
There is a PTAS for $P|d_j=d|c \cdot p_{sum}+Y$.
\end{corollary}

Notice that Theorem~\ref{thm:reduction} remains valid if we replace $Y$ by $c\cdot p_{sum}+Y$ in the \lmp\ and $\tilde{Y}$ by $c\cdot p_{sum} + \tilde{Y}$ in the minimization variant of the \rlp, thus we get the following:
\begin{corollary}
There is a PTAS for the \rlp~$P|p_j=1|c\cdot p_{sum} + \tilde{Y}$.
\end{corollary}

The approximation schemes for the 4 distinct problems all rely on the PTAS for the early work maximization problem, which is extended to the other 3 problems by appropriate transformations. In the design of the PTAS for the early work maximization problem, we had some difficulties in showing the approximation guarantee. The technique we found may be used for designing (fully) polynomial time approximation schemes for completely different combinatorial optimization problems as well.
We illustrate the main ideas for a maximization problem $\Pi$.
Suppose we have devised a family of algorithms $\{A_\varepsilon\}_{\varepsilon>0}$ for $\Pi$, but we are able to prove that it is a factor $(1-\varepsilon)$  approximation algorithm only under the hypothesis that $OPT(I) \geq \varepsilon f(I)$ for a problem instance $I$, where $f$ is a function assigning some rational number to $I$. Then we have to devise another algorithm, which is also a factor $(1-\varepsilon)$ approximation algorithm on those instances such that $OPT(I) < \varepsilon f(I)$.
Now, if we run both methods on an arbitrary instance $I$, then at least one of them will return a solution of value at least $(1-\varepsilon)$ times the optimum. Clearly, the combined method is an (F)PTAS for the problem $\Pi$.

\section{Equivalence of the \lmp~and the \rlp}\label{sec:res_lev}

\begin{proof}[Proof of Theorem \ref{thm:reduction}]
The proof consists of two parts.
First, we define a bijective function between the set of instances of the \lmp~and the set of instances of the \rlp~with unit-time jobs.
Then, we consider an arbitrary pair of instances of the two problems (the pair is determined by the previous function) and we define another bijective function between the schedules  of the two instances.

Consider an arbitrary instance $I$ of the \lmp~($m$ machines, $n$ jobs with processing times $p_j$ ($j\in \{1,\ldots,n\}$) and common due date $d$, and upper bound $N$ on the number of jobs on each machine).
The corresponding instance of the \rlp~has $N$ machines, $n$ jobs with processing times 1, resource requirements $a_j := p_j$ ($j=1,\ldots,n$), common deadline $C := m$,  and resource limit $L := d$. The proof of the theorem is divided into a series of claims, and the proofs of Claims~\ref{claim1}-\ref{claim3} can be found in the appendix.
\begin{claim}
The defined function is a bijection between the  sets of instances of the two problems.
\label{claim1}
\end{claim}

Now, we describe a mapping  from the set of feasible schedules of any instance of the \lmp\ to that of the corresponding instance of the \rlp.
Let instance $I$ of the \lmp~be fixed and let $I'$ be the corresponding instance of \rlp.
Let $S$ be any feasible schedule for the instance $I$, our function defines a schedule $S'$ for $I'$ based on $S$ as follows.
If a job $j$ is the $\ell^{th}$ job scheduled on machine $i$ in $S$ then schedule the corresponding job of $I'$  on machine $\ell$ at time $t_j(S'):=i-1$, for an illustration, see Fig.~\ref{fig:schmap}.

\begin{figure}
\begin{tikzpicture}[font=\small]

\def\ox{0} 
\def\oy{0} 
\coordinate(o) at (\ox,\oy); 

\def\ui{0}
\def\uii{4.2}
\coordinate(u1) at (\ui,\oy);
\coordinate(u2) at (\uii,\oy);

\tikzstyle{mystyle}=[draw, fill opacity=1, minimum height=0.5cm,rectangle, inner sep=0pt, font=\small]

\def\tl{6} 
\def\oyi{0}
\def\oyii{0.7}
\def\oyiii{1.4}
\draw [-latex](\ox,\oyi) node[above left]{$M_3$} -- (\ox+\tl,\oyi) node[above,font=\small]{$t$};
\draw [-latex](\ox,\oyii) node[above left]{$M_2$} -- (\ox+\tl,\oyii) node[above,font=\small]{};
\draw [-latex](\ox,\oyiii) node[above left]{$M_1$} -- (\ox+\tl,\oyiii) node[above,font=\small]{};

\coordinate (uq2) at (\uii,\oyi);
\draw[] (uq2) -- ($(uq2)-(0,0.2)$) node[below right=0cm and -0.3cm] {$d$}; 

\def\b{0.9}
\def\bb{0.7}
\def\bbb{1.5}
\def\bbbb{0.9}
\node(b1) [above right=(-0.01cm+\oyiii cm) and -0.01cm of o,mystyle, minimum 
width=0.7 cm]{1};
\node(b1)[right=0cm of b1,mystyle, minimum width=0.7 cm]{2};
\node(b1)[right=0cm of b1,mystyle, minimum width=1.4 cm]{3};
\node(b1)[right=0cm of b1,mystyle, minimum width=0.7 cm]{4};

\def\bi{0.65}
\def\bii{0.5}
\def\biii{0.7}
\def\biv{0.5}

\node(b1) [above right=(-0.01cm+\oyii cm) and -0.01cm of o,mystyle, minimum width=1.4 cm]{5};
\node(b1)[right=0cm of b1,mystyle, minimum width=0.7 cm]{6};
\node(b1)[right=0cm of b1,mystyle, minimum width=2.8 cm]{7};

\node(b1) [above right=-0.01cm and -0.01cm of o,mystyle, minimum width=0.7 cm]{8};
\node(b1)[right=0cm of b1,mystyle, minimum width=2.1 cm]{9};
\node(b1)[right=0cm of b1,mystyle, minimum width=0.7 cm]{10};
\node(b1)[right=0cm of b1,mystyle, minimum width=1.4 cm]{11};

\draw[dashed] (uq2) -- ($(uq2)+(0,1.9)$); 

\draw  (2,2.1) node[above]{$N=4$};

\def\ox{7.5} 
\def\oy{0} 
\def\pi{0.5}
\coordinate(o) at (\ox,\oy); 
\def\uiv{\ox+\pi+\pi+\pi}
\coordinate(u4) at (\uiv,\oy);

\def\tl{2.5} 
\def\oyi{0}
\def\oyii{\pi+0.2}
\def\oyiii{\pi+0.2+\pi+0.2}
\def\oyiv{\pi+0.2+\pi+0.2+\pi+0.2}
\draw [-latex](\ox,\oyi) node[above left]{$M'_4$} -- (\ox+\tl,\oyi) node[above,font=\small]{$t$};
\draw [-latex](\ox,\oyii) node[above left]{$M'_3$} -- (\ox+\tl,\oyii) node[above,font=\small]{};
\draw [-latex](\ox,\oyiii) node[above left]{$M'_2$} -- (\ox+\tl,\oyiii) node[above,font=\small]{};
\draw [-latex](\ox,\oyiv) node[above left]{$M'_1$} -- (\ox+\tl,\oyiv) node[above,font=\small]{};


\coordinate (C) at (\uiv,\oyi);
\draw[] (C) -- ($(C)-(0,0.2)$) node[below right=0cm and -0.3cm] {$C=3$}; 

\node(b1) [above right=(2.09 cm) and -0.01cm of o,mystyle, minimum width=\pi cm]{1};
\node(b1)[right=0cm of b1,mystyle, minimum width=\pi cm]{5};
\node(b1)[right=0cm of b1,mystyle, minimum width=\pi cm]{8};

\node(b1) [above right=(1.39 cm) and -0.01cm of o,mystyle, minimum width=\pi cm]{2};
\node(b1)[right=0cm of b1,mystyle, minimum width=\pi cm]{6};
\node(b1)[right=0cm of b1,mystyle, minimum width=\pi cm]{9};

\node(b1) [above right=(0.69 cm) and -0.01cm of o,mystyle, minimum width=\pi cm]{3};
\node(b1)[right=0cm of b1,mystyle, minimum width=\pi cm]{7};
\node(b1)[right=0cm of b1,mystyle, minimum width=\pi cm]{10};

\node(b1) [above right=-0.01cm and -0.01cm of o,mystyle, minimum width=\pi cm]{4};
\node(b1)[right=\pi cm of b1,mystyle, minimum width=\pi cm]{11};
\end{tikzpicture}
\caption{Corresponding schedules for \lmp~and \rlp.}\label{fig:schmap}
\end{figure}

\begin{claim}
$S'$ is feasible for $I'$.
\label{claim2}
\end{claim}

\begin{claim}
The mapping between the schedules for $I$ and that for $I'$ is a bijection.
\label{claim3}
\end{claim}

\begin{claim}
If the late work of some schedule $S$ for instance $I$ is $Y$, then the objective function value of the corresponding schedule $S'$ for $I'$ is also $Y$.
\end{claim}
\begin{proof}
Consider the $i^{th}$ machine $M_i$ ($i \in \{1,\ldots,m\}$) in $S$, let $\J_i$ denote the set of jobs scheduled on $M_i$ in $S$.
The late work on $M_i$ is $\max\{0,\sum_{j\in\J_k}p_j-d\}$, thus $Y=\sum_{k=1}^m \max\{0,\sum_{j\in\J_k}p_j-d\}$.
On the other hand, observe that the jobs of $\J_i$ are mapped to those jobs of the \rlp~that start at time point $i-1$ in $S'$.
The total resource requirement of these jobs  exceeds $L$ by  $\max\{0,\sum_{j\in\J_k}a_j-L\}$, thus the objective function value of $S'$ is $\sum_{i=1}^{C} \max\{0,\sum_{j\in\J_i}a_j-L\}=\sum_{i=1}^{m}\max\{0,\sum_{j\in\J_k}p_j-d\}=Y$,
since $L = d$, $C = m$, and $p_j= a_j$ by the mapping defined above.
\end{proof}
The above claims prove the theorem.
\end{proof}

\section{Inapproximability of $P2|d_j=d|c'+Y$}\label{sec:inapprox}
In this section we prove Theorem \ref{thm:inapprox}.

\begin{proof}[Proof of Theorem \ref{thm:inapprox}]
Let $c'$ be a fixed  positive rational number, and $\varepsilon > 0$  an arbitrarily small positive number. We show that if there is a polynomial time $\left(\frac{c'+1}{c'}-\varepsilon\right)$-approximation algorithm for $P2|d_j=d|c'+Y$ then we can decide in polynomial time any instance of the  PARTITION problem, which is an NP-hard decision problem \cite{garey1979computers}.
The latter problem is as follows:
\vskip 5pt
\noindent PARTITION: Given a set of $n$ items with positive integer item sizes $e_1,\ldots,e_n$, and one more positive integer $E$ such that $\sum_{i=1}^n e_i:=2E$. Question: does there exist a subset $H$ of the items such that $\sum_{i \in H} e_i:=E$?
\vskip 5pt

Consider an arbitrary instance of PARTITION, the corresponding instance of $P2|d_j=d|c'+Y$ has 2 machines, $n$ jobs with processing times $p_j:=e_j$, $j=1,2,\ldots,n$, and common due date  $d:=E$.

\begin{claim}
The answer to the instance of PARTITION is `yes' if and only if there exists a schedule of objective function value $c'$ in the corresponding instance of $P2|d_j=d|c'+Y$.
\end{claim}
\begin{proof}
Consider any instance of PARTITION, and the corresponding instance of $P2|d_j=d|c'+Y$.
Suppose the PARTITION problem instance admits a solution, i.e., there is a subset  $H$ of the items such that $\sum_{i\in H} e_i = E$.
In the corresponding instance of $P2|d_j=d|c'+Y$, schedule the jobs corresponding to the items in $H$ on the first machine, and the remaining jobs on the second machine in any order, without idle times. Then on both machines all jobs finish by $E$, thus $Y = 0$. Hence,  the value of the objective function $c'+Y$ is $c'$.

Conversely, suppose the scheduling problem instance admits a schedule of objective function value $c'$, then there is no late work in this schedule, thus each machine is working in $[0,d]$, because $\sum_{j=1}^n p_j= \sum _{i=1}^n e_i=2d = 2E$.
This means that the sum of the items that correspond to jobs scheduled on the first machine is $d=E$, thus the answer to the instance of PARTITION is 'yes'.
\end{proof}

Since all of the job processing times are integer numbers, there exists an optimal schedule such that all the jobs start at integer time points, and thus $Y \in \mathbb{Z}_{\geq 0}$.
Suppose there exists a $\left(\frac{c'+1}{c'}-\varepsilon\right)$-approximation algorithm for $P2|d_j=d|c'+Y$, then we can  decide in polynomial time with this algorithm whether the optimum value of an instance of our scheduling problem is $c'$ or at least $c'+1$. Therefore, we can decide in polynomial time the corresponding instance of PARTITION, which is impossible unless $P=NP$.
\end{proof}

\section{A PTAS for $P|d_j = d, n_i \leq N|X$} 
\label{sec:earlyPTAS}
 In this section we describe a PTAS for $P|d_j = d, n_i \leq N|X$. Note that the machine capacity $N$ is a positive integer such that $m \cdot N \geq n$, where $n$ is the number of the jobs, and $m$ is the number of identical parallel machines.
In fact, we will devise two algorithms (both paramterized by $\varepsilon$), and we will run both of them on the same input, and finally, we will choose the better of the two schedules obtained as the output of the algorithm.
After some preliminary observations, we will describe the two algorithms along with the proofs of their soundness, and in the end we combine them to prove Theorem~\ref{thm:ptas}.

Throughout this section,  $S^*$ denotes an optimal schedule for an instance of $P|d_j = d, n_i \leq N|X$.

\subsection{Family of algorithms for the case $X(S^*) \geq \varepsilon\cdot m \cdot d$}
\label{sec:earlyXbig}
In this section we describe a family of algorithms $\{\mathcal{A}_\varepsilon\ |\ \varepsilon > 0\}$, such that $\mathcal{A}_\varepsilon$ is a factor $(1-4\varepsilon)$ approximation algorithm  for the problem $P|d_j = d, n_i \leq N|X$ under the condition $X(S^*) \geq \varepsilon \cdot m \cdot d$.

We start by observing that if a job starts after $d$ then we do not have to deal with its exact starting time and with its machine assignment, because the total processing time of this job is late work.
We  can schedule these jobs from any time point after $d$ on any machine where we do not violate the machine capacity constraints.

Let $\varepsilon> 0$ be fixed.
We divide the set of jobs into three subsets, huge, big and small.
The set of {\em huge jobs\/} is
$\mathcal{H} := \{ j \in \J\ |\ p_j \geq d\}$, the set of {\em big jobs\/} is $\mathcal{B} := \{ j \in \J\ |\ \varepsilon^2 d \leq p_j < d\}$, and the remaining jobs are {\em small\/}. 

\begin{prop}
If there are at least $m$ huge jobs, then scheduling $m$, arbitrarily chosen huge jobs on $m$ distinct machines, and the rest of the jobs arbitrarily, yields an optimal schedule both for the maximum early work  and the minimum late work objectives.\label{prop:huge1}
\end{prop}
\begin{proof}
Let $S'$ be the schedule constructed as described in the statement of the proposition.
Then  $X(S') = m\cdot d$, which is the maximum possible early work. By equation (\ref{eq:lateXY}), $S'$ has minimum late work as well, thus it is optimal for both objective functions.
\end{proof}

\begin{prop}
If $|\mathcal{H}| \leq m-1$, then there exists an optimal schedule for the maximum early work as well as for the minimum late work objectives such that the huge jobs are scheduled on $|\mathcal{H}|$ distinct machines.\label{prop:huge2}
\end{prop}
\begin{proof}
Let $S^*$ be an optimal schedule for the early work (as well as for the late work) objective with the maximum number of machines on which a huge job is scheduled. Indirectly, suppose less than $|\mathcal{H}|$ machines process at least one huge job, hence, there exists a machine $M_1$ processing at least two huge jobs, say $j_1$ and $j_2$, in this order.
Since there are at most $m-1$ huge jobs, there exists a machine $M^*$ (in fact there are at least two), which does not process any huge jobs. 
If less than $N$ jobs are scheduled on $M^*$, then move job $j_2$ from $M_1$ to $M^*$, otherwise swap job $j_2$ with any of the jobs scheduled on $M^*$, and let $S'$ be the resulting schedule.
Clearly, the machine capacities are respected by $S'$, and both of the machines $M^*$ and $M_1$ work in the period $[0,d]$ in $S'$, while the work assigned to any other machine is the same in both schedules.
Hence, $X(S') \geq X(S^*)$. Therefore, $S'$ is optimal for the early work objective, and by equation (\ref{eq:lateXY}), for the late work objective as well. However, in $S'$ more machines process at least one huge job than in $S^*$, a contradiction.
\end{proof}

From now on, we assume that there are at most $m-1$ huge jobs, and we fix an optimal schedule $S^*$ in which the huge jobs are scheduled on distinct machines.

Our algorithm has three main phases: first, we schedule  all of  the huge jobs, and some of the big jobs such that they get a starting time smaller than $d$, then we schedule some of the small jobs such that they get a starting time smaller than $d$, and finally, we schedule the remaining big and small jobs, if any, arbitrarily while respecting the machine capacity constraints.

For each big job $j$ we round down its processing time $p_j$ to the greatest integer $p'_j:=\lceil \varepsilon^2 d (1+\varepsilon)^k \rceil$ ($k\in \mathbb{Z}$) such that $p'_j\leq p_j$.
Since we have $\varepsilon^2 d \leq p_j < d$ for each big job $j$, the number of the different $p'_j$ values is bounded by the constant $k_1:=\lfloor\log_{1+\varepsilon}(1/\varepsilon^2)\rfloor+1$ that depends on the fixed $\varepsilon$ only.
Let $\B_1,\B_2,\ldots,\B_{k_1}$ denote the  sets of the big jobs with the same rounded processing times, i.e., $\B_h:=\{j\in \J: p'_j=\lceil (1+\varepsilon)^{h-1}\cdot\varepsilon^2 d \rceil\}$ ($\B_h=\emptyset$ is possible).

For each machine without a huge job, we guess the number of the big jobs from each set $\mathcal{B}_h$ that start before $d$. 
This guess can be described by an {\em assignment\/} $A$, which consists of $k_1$ numbers $(\gamma_1,\gamma_2,\ldots,\gamma_{k_1})$, where $\gamma_h$ describes the number of the jobs from $\B_h$.     
A big job assignment $(\gamma_1,\gamma_2,\ldots,\gamma_{k_1})$ is {\em feasible},  if it does not violate the constraint on the number of the jobs on a machine, i.e., $\sum_{h=1}^{k_1} \gamma_h\leq N$, and all the selected jobs can be started before $d$, i.e., scheduling them in non-decreasing processing time order, each of the selected big job starts before $d$.
Let $k_2$ be the number of possible big job assignments. Since  the total number of big jobs that may start before $d$ on a machine is at most $\lfloor 1/\varepsilon^2\rfloor$, we have $k_2 \leq k_1^{\lfloor 1/\varepsilon^2\rfloor}$.
Let $A_1,A_2,\ldots, A_{k_2}$ denote the different feasible big job assignments.

A {\em layout\/} is a $k_2$ tuple $(t_1,t_2,\ldots,t_{k_2})$ that  specifies for each feasible assignment the number of the  machines that uses it. Let $\gamma_{ih}$ denote the number of big jobs from $\mathcal{B}_h$ assigned by $A_i$.
A layout is {\em feasible} if and only if $\sum_{i=1}^{k_2} t_i \gamma_{ih} \leq |\mathcal{B}_h|$ for each $h = 1,\ldots,k_1$.
 The number of feasible tuples is bounded by the number of non-negative, integer solutions of the inequality $\sum_{i=1}^{k_2} t_i\leq m-|\mathcal{H}|$, which is bounded by $\binom{m-|\mathcal{H}|+k_2}{k_2}$, a polynomial in the size of the input, since $k_2$ is a constant (that depends on $\varepsilon$ only).
In  Algorithm $A$, we examine each big job layout and get a complete schedule for each of them.

\medskip

\textbf{Algorithm $\mathbf{A}$}
\begin{enumerate}
\item Determine the set of feasible layouts.
\item For each layout $t$, perform the steps \ref{alg_a:assign}--\ref{alg_a:rem}.\label{alg_a:loop}
\item Assign the huge jobs of $\mathcal{H}$ to machines $M_1\ldots,M_{|\mathcal{H}|}$ arbitrarily, and big jobs to the remaining $m-|\mathcal{H}|$ machines according to $t$ ($t_i$ machines use assignment $A_i$)\label{alg_a:assign}
\item On each machine, schedule the assigned jobs from time point 0 on in arbitrary order.
\item If $N\geq n$, then invoke Algorithm $B$, otherwise invoke Algorithm $C$ to schedule small jobs. \label{alg_a:invoke}
\item Schedule the remaining jobs (small and big, if any) on the machines arbitrarily such that no machine receives more than $N$ jobs  in total (including the pre-assigned huge and big jobs).\label{alg_a:rem}
\item Output  $S_A$, which is the best schedule found in steps (\ref{alg_a:loop})-(\ref{alg_a:rem}).
\end{enumerate}

Now we turn to Algorithms $B$ and $C$ for scheduling small jobs. Algorithm $B$ is a simple greedy method which works only if there are no machine capacity constraints,  i.e., $N \geq n$.
 
\medskip

\textbf{Algorithm $\mathbf{B}$}

Input: partial schedule of big jobs
\begin{enumerate}
\item For $i=1,\ldots,m$ do:
\item Schedule a maximal subset of small jobs on machine $M_i$ after the big jobs without idle time such that no small job finishes after $d$.\label{alg_b:sch}
\end{enumerate}

Observe that the above method may assign a lot of small jobs to a machine, thus it may not yield a feasible schedule if  $N<n$ . 

Algorithm $C$ is much more complicated. 
Let  $\J^{small}$ denote the set of small jobs, $P^{small}_i\geq 0$ the idle time on machine $i$ before $d$, and $n^{small}_i$  the number of the jobs that can be scheduled on machine $i$ after the partial schedule of big jobs, i.e., $n^{small}_i$ is the difference between $N$ and the number of the big jobs assigned to machine $M_i$. Note that $P_i^{small} = 0$ if a huge job is assigned to machine $M_i$.

Our goal is to maximize the early work of the small jobs for a fixed assignment of big and huge jobs.
To simplify our problem, we only want to maximize  the total processing time of the small jobs that a machine completes before $d$.
This may decrease the objective function value of the final schedule, but we will show that this error is negligible. 

We can model the above problem with an integer program.
We introduce $n\cdot(m+1)$ binary variables $x_{ij}$ ($i=0,1,2,\ldots,m$, $j=1,2,\ldots,n$), where $x_{0,j}=1$ means that we do not schedule job $j$ to any machine before $d$, while in case of $1\leq i\leq m$, $x_{i,j}=1$ means that job $j$ will be scheduled on machine $i$, and will be completed not later than $d$.
\begin{align}
\max &\sum_{i=1}^m\sum_{j\in\J^{small}} x_{i,j}p_{j}\label{eq:obj}\\
&s.t.\notag\\
&\sum_{j\in \J^{small}} x_{i,j}p_j\leq P^{small}_i,&&  i=1,\ldots,m,
\end{align}
\begin{align}
&\sum_{j\in \J^{small}} x_{i,j}\leq n^{small}_i,& & i=1,\ldots, m,\\
&\sum_{i=0}^m x_{i,j}= 1,&&  j\in \J^{small},\\
&x_{i,j}\in\{0,1\}, && i=0,\ldots,m,\;j\in\J^{small}.\label{eq:ip}
\end{align}
We  get the LP-relaxation of the above integer program by replacing  $x_{i,j}\in\{0,1\}$ with $x_{i,j}\geq 0$ in the constraints (\ref{eq:ip}).

\medskip

\textbf{Algorithm $\mathbf{C}$}

Input: partial schedule of big jobs
\begin{enumerate}
\item Determine the values $P^{small}_i$, $n^{small}_i$ for $i=1,\ldots,m$.
\item Solve the LP-relaxation of  (\ref{eq:obj})--(\ref{eq:ip}), and let $\bar{x}$ be a basic optimal solution.
\item For $i=1,\ldots,m$, if $\bar{x}_{i,j}=1$ for a job $j$, then assign that job to machine $i$.
\item For each machine, schedule the assigned jobs right after the big jobs without idle times in arbitrary order.
\end{enumerate}
Observe that fractional jobs of the optimal LP solution are not assigned to any machine by Algorithm $C$, but they will be scheduled by the Step~\ref{alg_a:rem} of Algorithm $A$.

The proofs of the following two claims easily follow from the definitions.
\begin{prop} \label{prop:feas}
$S_A$ is feasible.
\end{prop}

\begin{prop}\label{prop:run_time_B}
The  time complexity of Algorithm $B$ is polynomially bounded in the size of the input.
\end{prop}

\begin{prop}\label{prop:run_time_C}
The  time complexity of Algorithm $C$ is polynomially bounded in the size of the input.
\end{prop}
\begin{proof}
We can determine a basic solution of a linear  program with $nm$ variables and $n+2m$ constraints in two steps.
First, apply a polynomial time interior-point algorithm to find a pair of primal-dual optimal solutions, and then, we can use Megiddo's method to determine a basic solution $\bar{x}$ for the primal program, see e.g., \citet{wright1997primal}.
The other steps of Algorithm $C$ require linear time.
\end{proof}

\begin{prop}\label{prop:run_time_A}
The time complexity of Algorithm $A$ is polynomially bounded in the size of the input.
\end{prop}
\begin{proof}
Recall that the number of the feasible layouts is polynomial (at most $\binom{m+k_2}{k_2}$). 
Each of the steps \ref{alg_a:assign}-\ref{alg_a:rem} requires  $O(n m)$ time, except Step \ref{alg_a:invoke} if it invokes Algorithm $C$, but it is also polynomial due to Proposition \ref{prop:run_time_C}. 
\end{proof}

Without loss of generality, we assume that in  $S^*$ the huge and big jobs precede the small jobs on each machine, and the big jobs are scheduled in non-decreasing processing time order on each machine. 
We introduce an intermediate schedule $S_{int}$: it is the same as $S^*$ except that the processing time of each big job is rounded as in Algorithm $A$. That is, the processing time of each big job is rounded down to the greatest number of the form $\lceil \varepsilon^2 d (1+\varepsilon)^k\rceil$, ($k\in \mathbb{Z}$), and after rounding we re-schedule the jobs on each machine in the same order as in $S^*$, but with the decreased processing times of the big jobs. 
By considering those big jobs on the machines that start before $d$ in $S_{int}$, we can uniquely identify an assignment of big jobs for each machine. Therefore, we can determine the layout $t^*$ of the big jobs that start before $d$ in $S_{int}$. 
Now we state and prove the main result of this section.
\begin{theorem}
If $X(S^*) \geq \varepsilon \cdot m\cdot d$, then   Algorithm $A$ is a factor $(1-4\varepsilon)$ approximation algorithm for $P|d_j = d,n_i \leq N|X$.\label{thm:earlyXbig}
\end{theorem}
\begin{proof}
Recall that $S_{int}$ is the schedule obtained from $S^*$ by rounding down the processing time of each big job, and shifting the jobs to the left, if necessary, to eliminate any idle times (created by rounding) on the machines.
Since $p_j / (1+\varepsilon) < p'_j \leq p_j$, we have $X(S_{int}) \geq X(S^*)/(1+\varepsilon) \geq (1-\varepsilon)X(S^*)$.
Let $t^*$ be the layout of big jobs corresponding to $S_{int}$.
Algorithm $A$ will consider the layout $t^*$ at some iteration, and let $S$ be the schedule created from  $t^*$.
Since $X(S_A) \geq X(S)$, it suffices to prove that $X(S) \geq (1-4\varepsilon)X(S^*)$.
To achieve this, we proceed by proving a series of lemmas.
\begin{lemma}
If $N\geq n$ and $X(S^*) \geq \varepsilon \cdot m \cdot d$, then $X(S)\geq (1-\varepsilon)X(S^*)$.
\label{lem:X_N_big}
\end{lemma}
\begin{proof}
If Algorithm $B$ schedules all the small jobs when creating schedule $S$, then the only jobs finishing after $d$ can be big and huge jobs. Since the set of big and huge jobs that start before $d$ in schedule $S$ contains all the big and huge jobs that start before $d$ in schedule $S_{int}$, we get $X(S) \geq X(S_{int})$.

If there is at least one small job that remains unscheduled  by Algorithm $B$, then 
consider the early work in $S$. We know that the total processing time on each machine is at least $d(1-\varepsilon^ 2)$ due the the condition of Step~\ref{alg_b:sch} of Algorithm $B$.
Hence, $X(S) \geq md(1-\varepsilon^2)$.
Since $X(S) \leq X(S^*) \leq m\cdot d$, and $X(S^*) \geq \varepsilon\cdot m\cdot d$ by assumption, we derive
\[
X(S) \geq (1-\varepsilon^2) d \cdot m \geq (1-\varepsilon)X(S^*),
\]
as claimed.
\end{proof}

\begin{prop}
If $N<n$, then $X(S) \geq X(S_{int})-3\varepsilon^2\cdot d \cdot m$.
\label{prop:XS_Xint}
\end{prop}
For a proof, see the Appendix.
\begin{lemma}
If $N < n$ and $X(S^*) \geq \varepsilon \cdot m \cdot d$, then $X(S) \geq (1-4\varepsilon)X(S^*)$.
\label{lem:X_N_small}
\end{lemma}
\begin{proof}
By Proposition~\ref{prop:XS_Xint}, $X(S) \geq X(S_{int})-3\varepsilon^2\cdot d \cdot m$.
Therefore, using the assumption of the lemma, we derive
\[
X(S) \geq X(S_{int})-3\varepsilon^2\cdot d \cdot m  \geq X(S^*)(1-\varepsilon) - 3\varepsilon X(S^*) = (1-4\varepsilon) X(S^*).
\]
\end{proof}
Now we can finish the proof of Theorem~\ref{thm:earlyXbig}.
We have proved that Algorithm $A$ creates a feasible schedule $S_A$ (Proposition \ref{prop:feas}) in polynomial time (Proposition \ref{prop:run_time_A}) such that $X(S_A) \geq (1-4\varepsilon) X(S^*)$ (Lemmas \ref{lem:X_N_big}-\ref{lem:X_N_small}), thus the theorem is proved.
\end{proof}
Theorem~\ref{thm:earlyXbig} has a strong assumption, namely, $X(S^*) \geq \varepsilon\cdot m\cdot d$.
In the next section, we describe a complementary method, which works if $X(S^*) < \varepsilon\cdot m\cdot d$.

\subsection{The second approximation algorithm}
\label{sec:earlySmall}
We will show that if $X(S^*) < \varepsilon\cdot m\cdot d$, then
scheduling the jobs in non-increasing processing time order by list-scheduling  while respecting the capacity constraints of the machines yields an approximation Algorithm $B$oth for minimizing the late work and for maximizing the early work as well.
Recall the {\em list-scheduling\/} method of \citet{graham1969bounds} for scheduling jobs on parallel machines. It processes the jobs in a given order, and it always schedules the next job on the least loaded machine. 
In order to take into account the capacity constraints of the machines, we will use the following variant of list-scheduling. 

\medskip

\textbf{Algorithm LS}

Input: list of jobs, number of machines $m$, and common machine capacity $N$.
\begin{enumerate}
\item Let $n_i := 0$, and $L_i:= 0$ for $i=1,\ldots,m$.
\item Process the jobs in the order given by the list. When processing the next job $j$ from the list, choose the machine with minimum $L_i$ value  among those machines with $n_i < N$, and break ties arbitrarily. Let $i$ be the index of the machine chosen. Then set $t_j(S_{LPT}) = L_i$, $\mu_j(S_{LPT}) := i$, $L_i := L_i + p_j$ and $n_i := n_i+1$.
\item Return $S_{LPT}$.
\end{enumerate}
Let $S_{LPT}$ be the schedule obtained by list-scheduling for machines with capacities using the above job order.

\begin{theorem}
If $X(S^*) < \varepsilon \cdot m \cdot d$ and $\varepsilon \leq 1/3$, then $X(S_{LPT}) \geq (1-2\varepsilon) X(S^*)$ and $c\cdot p_{sum} + Y(S_{LPT}) \leq  (1+2\varepsilon/c)(c\cdot p_{sum} +  Y(S^*))$.\label{thm:earlyXsmall}
\end{theorem}
\begin{proof}
First, we prove $X(S_{LPT}) \geq (1-2\varepsilon) X(S^*)$, and then we derive from it the second statement of the theorem.
Since $X(S^*) \leq \varepsilon \cdot m \cdot d$, there can be at most $m-1$ jobs of processing time at least $\varepsilon d$.
Since $X(S_{LPT}) \leq X(S^*)$, we can also deduce that in $S_{LPT}$ there is a machine on which the total processing time of the jobs is less than $\varepsilon d$.

First suppose that all jobs start {\em before\/} $\varepsilon d$ in $S_{LPT}$. 
Since there are  $k\leq m-1$ jobs of processing time at least $\varepsilon d$, all these {\em long jobs\/}  start on distinct machines in $S_{LPT}$, since these are the longest $k$ jobs. All the remaining  jobs have a processing time smaller than $\varepsilon d$, and they are scheduled on the remaining $m-k$ machines. Therefore, the work finishes by time $2\varepsilon d$ on the remaining machines. Since $\varepsilon\leq 1/3$, the jobs, if any, that do not finish before $d$ in $S_{LPT}$ must be long jobs. Since the long jobs are scheduled on distinct machines in $S_{LPT}$, there is no way to decrease the late work of this schedule, or equivalently, to increase the early work, thus, $S_{LPT}$ must be optimal for both objectives.

Now suppose there is a job $j$ which starts at or after $\varepsilon d$ in $S_{LPT}$. 
Then there is a machine $M^*$ in $S_{LPT}$ with $N$ jobs and the total processing time of these jobs is smaller  than $\varepsilon d$, otherwise either job $j$ could be scheduled on $M^*$ (which would contradict the rules of the list-scheduling algorithm), or $X(S_{LPT}) \geq \varepsilon \cdot m \cdot d$ (which would contradict the assumption $X(S^*) < \varepsilon \cdot m \cdot d$, since $S_{LPT}$ is a feasible schedule, and $S^*$ is an optimal schedule, thus
$\varepsilon \cdot m \cdot d \leq X(S_{LPT}) \leq X(S^*)$).

We claim that on any machine, the total processing time of those jobs that start at or after $\varepsilon d$ is at most $\varepsilon d$. This is so, because the jobs are scheduled in non-increasing processing time order, and no machine may receive more than $N$ jobs. Consequently, if a job is started at or later than $\varepsilon d$ on some machine, it has  a processing time not greater than the shortest processing time on $M^*$. Hence, the total processing time of the jobs scheduled on $M^*$  is indeed an upper bound on the total processing time of those jobs started at or later than $\varepsilon d$ on any single machine. 

By our claim, if there are only short jobs (of processing time smaller than $\varepsilon d$) on a machine, then the total work assigned to it by $S_{LPT}$ is at most $3\varepsilon d$. Hence, all these jobs finish by $d$, since $\varepsilon \leq 1/3$.
Consequently, if a job finishes after $d$ in $S_{LPT}$, then it must be scheduled on a machine with a long job. Let $g$ be the number of those machines on which some job is late, i.e., finishes after $d$ in $S_{LPT}$. Consider any of these $g$ machines. It has a long job scheduled first, and then some short jobs. The total processing time of these short jobs is at most $\varepsilon d$, since each of them starts after $\varepsilon d$. Hence, the late work can be decreased by at most $g\cdot  \varepsilon d$ by scheduling some of  the short jobs early in a more clever way than in $S_{LPT}$. 
Consequently, $X(S_{LPT}) + g \cdot \varepsilon d \geq X(S^*)$. 

Now, we bound $g d$. As we have observed, if a machine has some late work on it in $S_{LPT}$, then it has a long job, and some short jobs of total processing time at most $\varepsilon d$. Hence, the length of the long job must be at least $d(1-\varepsilon)$. Therefore, $X(S^*) \geq gd(1-\varepsilon)$.
Using this observation, we obtain the first statement: 
\[
X(S_{LPT}) \geq X(S^*) - \varepsilon \cdot gd \geq X(S^*)- \varepsilon X(S^*)/(1-\varepsilon) \geq X(S^*)(1-2\varepsilon),
\]
where the last inequality follows from $\varepsilon / (1-\varepsilon) \leq 2\varepsilon$ if $0 < \varepsilon \leq 1/2$.

Now we derive the second statement of the theorem.
By equation (\ref{eq:lateXY}), $Y(S_{LPT}) = p_{sum} - X(S_{LPT})$.
Hence, we compute
\[
\begin{split}
Y(S_{LPT}) & = p_{sum} - X(S_{LPT}) \leq p_{sum} - X(S^*)(1-2\varepsilon)\\
& = p_{sum} - (p_{sum}-Y(S^*))(1-2\varepsilon) \\
& = p_{sum} - (p_{sum} - 2\varepsilon p_{sum} - Y(S^*) + 2\varepsilon Y(S^*))\\
& \leq Y(S^*) + 2\varepsilon p_{sum}.
\end{split}
\]
To finish the proof, observe that
\[
c\cdot p_{sum} + Y(S_{LPT}) \leq c\cdot p_{sum} + Y(S^*) + 2\varepsilon p_{sum} \leq (1+2\varepsilon/c)(c\cdot p_{sum} + Y(S^*)).
\]
\end{proof}

\subsection{The combined method}
In this section we combine the methods of Section~\ref{sec:earlyXbig} and Section~\ref{sec:earlySmall} to get a PTAS for $P|d_j = d, n_i \leq N|X$.

\begin{proof}[Proof of Theorem~\ref{thm:ptas_ew}]
By Theorems~\ref{thm:earlyXbig} and \ref{thm:earlyXsmall}, the following algorithm is a PTAS for $P|d_j = d, n_i \leq N|X$.

\medskip

\textbf{Algorithm PTAS}
 
Input: problem instance and parameter $0 < \varepsilon \leq 1/3$.
\begin{enumerate}
\item Run Algorithm $\mathrm{A}$ and let $S_{A}$ the best schedule found.
\item Run Algorithm LS with non-increasing processing time order of the jobs, and let $S_{LPT}$ be the schedule obtained. 
\item If $X(S_{A}) \geq X(S_{LPT})$, then output $S_{A'}$, else output $S_{LPT}$.
\end{enumerate}
Since the conditions of Theorems~\ref{thm:earlyXbig} and \ref{thm:earlyXsmall} are complementary, it follows that Algorithm PTAS always outputs a solution of value at least $(1-4\varepsilon)$ times the optimum.
The time complexity in either case is polynomial in the size of the input, hence, the algorithm is indeed a PTAS for our scheduling problem.
\end{proof}

\section{A PTAS for $P|d_j=d,n_i\leq N|c\cdot p_{sum}+Y$}\label{sec:ptas}
In this section we adapt the  PTAS of Section~\ref{sec:earlyPTAS} to the problem $P|d_j=d,n_i\leq N|c\cdot p_{sum}+Y$.
Throughout this section, $S^*$ denotes an optimal solution of a problem instance for the late work objective, and by equation (\ref{eq:lateXY}) for the early work objective as well.
\subsection{The first family of algorithms}
\label{sec:lateXbig}
In this section we describe a family of algorithms $\{\mathcal{A}_\varepsilon\ |\ \varepsilon > 0\}$, such that $\mathcal{A}_\varepsilon$ is a factor $(1+c_0\cdot \varepsilon)$ approximation algorithm  for the problem $P|d_j = d, n_i \leq N|c\cdot p_{sum}+Y$ under the condition $X(S^*) \geq \varepsilon \cdot m \cdot d$, where $c_0$ is a universal constant, independent of $\varepsilon$ and the problem instances.

Recall the definition of huge, big and small jobs from Section~\ref{sec:ptas}, we use the same partitioning of the set of jobs in this section as well.

By Propositions \ref{prop:huge1} and \ref{prop:huge2}, it suffices to consider the case when there are at most $m-1$ huge jobs.
However, in this section we round up the processing time $p_j$ of each big job $j$ to the smallest integer of the form $\lfloor \varepsilon^2 d(1+\varepsilon)^k\rfloor$, where $k\in \mathbb{Z}_{\geq 0}$.
Since $\varepsilon^2 d \leq p_j < d$ for each big job, there are at most $k_1 := \lfloor \log_{1+\varepsilon}1/\varepsilon^2\rfloor+1$ distinct rounded processing times of the big jobs.
Let $\B_1,\B_2,\ldots,\B_{k_1}$ denote the  sets of the big jobs with the same rounded processing times, i.e., $\B_h:=\{j\in \J: p'_j=\lfloor \varepsilon^2 d \cdot (1+\varepsilon)^{h-1}\rfloor\}$ ($\B_h=\emptyset$ is possible).
We also define the assignments of big jobs to machines and the layouts in the same way as in Section~\ref{sec:earlyPTAS}, but using the jobs classes $\B_h$ just defined.

\begin{theorem}
If $X(S^*) \geq \varepsilon \cdot m \cdot d$, then Algorithm $A$ is a factor $(1+4\varepsilon/c)$ approximation algorithm for $P|d_j = d, n_i \leq N| c \cdot p_{sum} + Y$.\label{thm:PTAS:YS*_big}
\end{theorem}
\begin{proof}
Let $S_{int}$ be the schedule obtained from $S^*$ by rounding up the processing time of each big job, and shifting the jobs to the right, if necessary, so that the jobs do not overlap on any machines.
Let $t^*$ be the layout of big jobs corresponding to $S_{int}$ (defined as in Section~\ref{sec:earlyPTAS}).
Algorithm $A$ will consider the layout $t^*$ at some iteration, and let $S$ be the schedule created from  $t^*$.
Since $Y(S_{A}) \leq Y(S)$, it suffices to prove that $c\cdot p_{sum}+Y(S) \leq (1+O(\varepsilon))(c\cdot p_{sum} + Y(S^*))$, and this is what we accomplish subsequently.
The claimed approximation factor is proved by a series of three lemmas.
\begin{lemma}\label{lem:S_int-S*}
$c\cdot p_{sum}+Y(S_{int})\leq (1+\varepsilon/c)(c\cdot p_{sum}+Y(S^*))$.
\end{lemma}
\begin{proof}
Observe that the rounding procedure increases the late work by at most $\varepsilon p_{sum}$ (recall that $p_{sum}:= \sum_{j\in\J} p_j$).
Hence, we have
\[
c\cdot p_{sum} + Y(S_{int}) \leq c\cdot p_{sum} + Y(S^ *) + \varepsilon p_{sum} \leq (1+\varepsilon/c) (c \cdot p_{sum} + Y(S^*).
\]
\end{proof}

\begin{lemma}\label{lem:S_int_N-inftty}
If $N\geq n$ and $X(S^*) \geq \varepsilon \cdot m \cdot d$, then $c\cdot p_{sum}+Y(S)\leq (1+2\varepsilon/c)(c\cdot p_{sum}+Y(S^*))$.
\end{lemma}
\begin{proof}
If Algorithm $B$ schedules all the small jobs when creating schedule $S$, then the only jobs finishing after $d$ can be big and huge jobs. Since the set of big and huge jobs that start before $d$ in schedule $S$ contains all the big and huge jobs that start before $d$ in schedule $S_{int}$, we get $Y(S) \leq Y(S_{int})$.

If there is at least one small job that remains unscheduled after Step \ref{alg_a:invoke} of Algorithm $A$, then consider the early work in $S$.
We know that the total processing time on each machine is at least $(1-\varepsilon^2)d$ due to the condition in Step \ref{alg_b:sch} of Algorithm $B$, thus $X(S)\geq (1-\varepsilon^2)d\cdot m$.
On the other hand, $X(S_{int})\leq d\cdot m$ is trivial, thus we have $Y(S)\leq Y(S_{int})+\varepsilon^2 d\cdot m$ due to (\ref{eq:lateXY}). Finally, we have
\[
\begin{split}
c \cdot p_{sum} + Y(S)& \leq c \cdot p_{sum} + Y(S_{int}) +\varepsilon^2 d\cdot m \\
& \leq c \cdot p_{sum} + Y(S_{int}) +\varepsilon X(S^*)\\
& \leq (1+\varepsilon/c)(c\cdot p_{sum}+Y(S^*)) + \varepsilon (p_{sum} - Y(S^*))\\
& \leq (1+2\varepsilon/c)(c\cdot p_{sum}+Y(S^*)),
\end{split}
\]
where the second inequality follows from the assumption  $X(S^*) \geq \varepsilon \cdot m \cdot d$, and the third  from Lemma~\ref{lem:S_int-S*} and equation (\ref{eq:lateXY}).
\end{proof}

\begin{lemma}\label{lem:S-S_int}
If $N<n$ and $X(S^*) \geq \varepsilon \cdot m \cdot d$, then $c\cdot p_{sum}+Y(S)\leq (1+4\varepsilon/c)(c\cdot p_{sum}+Y(S^*))$.
\end{lemma}
\begin{proof}
By Proposition~\ref{prop:XS_Xint} and equation (\ref{eq:lateXY}), we have $Y(S) \leq Y(S_{int}) + 3\varepsilon^2 dm$.
Therefore,
\[
\begin{split}
c\cdot p_{sum} + Y(S) & \leq  c\cdot p_{sum} + Y(S_{int}) + 3\varepsilon^2 dm \\
& \leq c\cdot p_{sum} + Y(S_{int}) + 3\varepsilon X(S^*) \\
& \leq (1+\varepsilon/c)(c\cdot p_{sum} + Y(S^*)) + 3\varepsilon(p_{sum} - Y(S^*))\\
& \leq (1+4\varepsilon/c)(c\cdot p_{sum} + Y(S^*)),
\end{split}
\]
where the second inequality follows from the assumption  $X(S^*) \geq \varepsilon \cdot m \cdot d$, and the third  from Lemma~\ref{lem:S_int-S*} and equation (\ref{eq:lateXY}).
\end{proof}
Now we can finish the proof of Theorem~\ref{thm:PTAS:YS*_big}.
We have proved that Algorithm $A$ creates a feasible schedule $S_A$ (Proposition \ref{prop:feas}) in polynomial time (Proposition \ref{prop:run_time_A}) such that $c\cdot p_{sum} + Y(S_A) \leq (1+4\varepsilon / c) (c \cdot p_{sum} + Y(S^*))$ (Lemmas \ref{lem:S_int-S*},  \ref{lem:S_int_N-inftty}, and \ref{lem:S-S_int}), thus the theorem is proved.
\end{proof}

\subsection{The combined method}
In this section we show how to combine the methods of Section~\ref{sec:earlySmall} and Section~\ref{sec:lateXbig} to get a PTAS for $P|d_j=d,n_i \leq N|c \cdot p_{sum} + Y$.
\begin{proof}[Proof of Theorem \ref{thm:ptas}]
By Theorems~\ref{thm:PTAS:YS*_big} and \ref{thm:earlyXsmall}, we propose the following algorithm for $P|d_j=d,n_i \leq N|c \cdot p_{sum} + Y$.

\medskip

\textbf{Algorithm PTAS}

Input: problem instance and parameter $\varepsilon \leq 1/3$.
\begin{enumerate}
\item Run Algorithm $A$ and let $S_{A}$ the best schedule found.
\item Run list-scheduling with non-increasing processing time order of the jobs, and let $S_{LPT}$ be the schedule obtained. 
\item If $Y(S_{A}) \leq Y(S_{LPT})$, then output $S_{A}$, else output $S_{LPT}$.
\end{enumerate}
Since the conditions of Theorems~\ref{thm:PTAS:YS*_big} and \ref{thm:earlyXsmall} are complementary, it follows that Algorithm PTAS always outputs a solution of value at most $(1+4\varepsilon/c)$ times the optimum.
The time complexity in either case is polynomial in the size of the input, hence, the algorithm is indeed a PTAS for our scheduling problem.
\end{proof}

\section{Final remark}

In this paper we have described a common approximation framework for 4 problems which have common roots.
However, there remained a number of open questions. For instance, is there a simple constant factor approximation algorithm for maximizing the early work which runs on arbitrary input, and has a running time suitable for practical applications?
The same question can be asked for the late work minimization problem with the objective $c+Y$ for some positive $c$.

\section{Appendix}

\begin{proof}[Proof of Claim~\ref{claim1}]
The function is injective (different instances of the \lmp~are mapped to different instances of the \rlp), and surjective (for every instance $I'$ of the \rlp~there is an instance $I$ of the \lmp~such that $I$ is mapped to $I'$), thus it is bijective.
\end{proof}

\begin{proof}[Proof of Claim~\ref{claim2}]
Since there are at most $N$ jobs scheduled on a machine in $S$, thus we assign each job to one of the $N$ machines of $I'$.
Furthermore, each job in $I'$ has a unit processing time, hence the jobs do not overlap. 
\end{proof}

\begin{proof}[Proof of Claim~\ref{claim3}]
It is easy to see that the given mapping of schedules is injective. Moreover, let $S'$ be any schedule for $I'$.
We define $S$ for $I$ such that $S$ is mapped to $S'$ as follows.
Suppose job $j$ starts on $M'_\ell$ at time point $i-1$ for some $i \in\{1,\ldots,C\}$ in $S'$, then  $j$ is the $\ell^{th}$ job on $\mu_j(S) = i$. Since in $S'$, there is no idle machine among $M'_1,\ldots,M'_\ell$ by definition, $S$ is feasible, and the value of $t_j(S)$ is well defined.
\end{proof}

\begin{proof}[Proof of Proposition~\ref{prop:XS_Xint}]
Consider Algorithm $C$, when it creates $S$.
It solves (\ref{eq:obj})--(\ref{eq:ip}) and $\bar{x}$ is the optimal basic solution that we get from the algorithm.
Recall that if $i\geq 1$ then $\bar{x}_{i,j}=1$ if and only if job $j$ is assigned to machine $i$ by Algorithm $C$.
We introduce another integer solution $x'$ of (\ref{eq:obj})--(\ref{eq:ip}).
Let $x'_{i,j}:=1$, if a small job $j$ completes before $d$ on machine $i$ in $S_{int}$, otherwise, $x'_{i,j}:=0$.
Note that $x'$ is a feasible solution, because $S_{int}$ is a feasible schedule.

Let $v(x)$ denote the objective function value of a solution $x$ of (\ref{eq:obj})--(\ref{eq:ip}),
$OPT_{IP}$ the optimum value of (\ref{eq:obj})--(\ref{eq:ip}) and $OPT_{LP}$ the optimum value of its linear relaxation.
For any feasible solution $x$ of (\ref{eq:obj})--(\ref{eq:ip}), we have $OPT_{LP}\geq OPT_{IP} \geq v(x)$.
Let $X^{small}_{int}$ denote the early work of the small jobs in $S_{int}$ and $X^{small}_{S}$ the same in $S$.
Observe that $v(x')$, which is the total early work of the small jobs that complete before $d$ in $S_{int}$, is at least $X^{small}_{int}-\varepsilon^2 dm$, because there is at most one small job  on each machine that starts before, and ends after $d$, and recall that each small job is shorter than $\varepsilon^2 d$.
Then
\[\begin{split}
X_{S}^{small} &\geq  v(\lfloor \bar{x}\rfloor) \geq OPT_{LP}-2\varepsilon^2 dm\geq OPT_{IP}-2\varepsilon^2 dm \geq v(x')-2\varepsilon^2 dm \\
& \geq X^{small}_{int}-3 \varepsilon^2 dm.
\end{split}
\]
The first inequality is trivial, while we have already proved the last three inequalities.
It  remained to prove the second inequality, i.e., $v(\lfloor \bar{x}\rfloor) \geq OPT_{LP}-2\varepsilon^2 dm$.
Let $e$ denote the number of the small jobs $j$ with $\bar{x}_{i,j}=1$ for some $i$  ($i=0,\ldots,m$) in Algorithm $C$, and $f:=n-e$ the number of the 'fractionally assigned' small jobs.
Note that for each of these small jobs, we have $i_1\neq i_2$ ($0\leq i_1,i_2\leq m$) such that $\bar{x}_{i_1,j},\bar{x}_{i_2,j}>0$).
Since $\bar{x}$ is a basic solution there are at most $n+2m$ non-zero values among its coordinates.
Hence, we have $e+2f\leq n+2m$, therefore, we have $f\leq 2m$.
To sum up, we have
\begin{align*}
OPT_{LP}=&\sum_{i=1}^m \left(\sum_{j:\bar{x}_{i,j}=1}p_j+\sum_{j \text{ frac.~assigned}}\bar{x}_{i,j}p_j\right)=\\
&v(\lfloor\bar{x}\rfloor)+ \sum_{j \text{ frac.~assigned}}p_j\sum_{i=1}^m\bar{x}_{i,j}\leq\\
&v(\lfloor\bar{x}\rfloor)+2\varepsilon^2m d,
\end{align*}
where the last inequality follows from $f\leq 2m$, from $p_j\leq \varepsilon^2 d$ for each small job $j$, and from $\sum_{i=1}^m\bar{x}_{i,j}\leq 1$.

Finally, observe that $X^{small}_S \geq X^{small}_{int} - 3\varepsilon^2 dm$ implies  $X(S) \geq X(S_{int})-3\varepsilon^2 dm$,  since the set of big and huge jobs that start before $d$ in $S$ contains those of schedule  $S_{int}$.
\end{proof}

\bibliographystyle{apa}
\bibliography{late_work}

\end{document}